\documentclass[12pt]{article}
\usepackage{amsmath, amssymb, amsthm, amsfonts}
\usepackage[margin=1 in]{geometry}
\usepackage{blindtext, xcolor}
\usepackage{algorithm}
\usepackage{algorithmicx}
\usepackage{algpseudocode}
\usepackage{apacite}
\usepackage{natbib}
\usepackage{appendix}
\usepackage{rotating}
\usepackage{hyperref}
\usepackage{float}
\usepackage{graphicx}
\usepackage{caption,subcaption}
\usepackage{cleveref}
\usepackage{arydshln}
\usepackage{authblk}
\usepackage{xr,bm}
\usepackage{mathrsfs}
\usepackage{booktabs}

\newcommand{\be} {\begin{eqnarray*}}
	\newcommand{\ee} {\end{eqnarray*}}

\theoremstyle{definition}

\newtheorem{theorem}{Theorem}
\newtheorem{lemma}{Lemma}

\def\*#1{\bm{#1}}

\title{Bayesian Neural Networks vs. Mixture Density Networks: Theoretical and Empirical Insights for Uncertainty-Aware Nonlinear Modeling}
\author[1]{Riddhi Pratim Ghosh}
\author[2]{Ian Barnett}

\affil[1]{Department of Mathematics and Statistics, Bowling Green State University}
\affil[2]{Department of Biostatistics, University of Pennsylvania}
\date{}

\begin{document}
\maketitle
\begin{abstract}
	This paper investigates two prominent probabilistic neural modeling paradigms: Bayesian Neural Networks (BNNs) and Mixture Density Networks (MDNs) for uncertainty-aware nonlinear regression. While BNNs incorporate epistemic uncertainty by placing prior distributions over network parameters, MDNs directly model the conditional output distribution, thereby capturing multimodal and heteroscedastic data-generating mechanisms. We present a unified theoretical and empirical framework comparing these approaches. On the theoretical side, we derive convergence rates and error bounds under Hölder smoothness conditions, showing that MDNs achieve faster Kullback–Leibler (KL) divergence convergence due to their likelihood-based nature, whereas BNNs exhibit additional approximation bias induced by variational inference. Empirically, we evaluate both architectures on synthetic nonlinear datasets and a radiographic benchmark (RSNA Pediatric Bone Age Challenge). Quantitative and qualitative results demonstrate that MDNs more effectively capture multimodal responses and adaptive uncertainty, whereas BNNs provide more interpretable epistemic uncertainty under limited data. Our findings clarify the complementary strengths of posterior-based and likelihood-based probabilistic learning, offering guidance for uncertainty-aware modeling in nonlinear systems.
\end{abstract}

\vspace{1cm}


\newpage
\section{Introduction} \label{intro}

Modeling complex, non-linear, and uncertain relationships between input and output variables remains a central challenge in modern statistical learning and artificial intelligence. Traditional neural networks, trained via point estimation, have demonstrated remarkable success in a variety of domains but inherently provide deterministic predictions - that is, single-valued outputs without accompanying measures of uncertainty. This limitation becomes critical in domains characterized by limited, noisy, or ambiguous data, such as medicine, climate science, or finance, where quantifying uncertainty is as important as producing accurate predictions \citep{gal2016dropout, kendall2017uncertainties, abdar2021review}.

Bayesian Neural Networks (BNNs) provide a probabilistic extension of standard neural networks by treating weights and biases as random variables endowed with prior distributions \citep{mackay1992practical, neal1996bayesian}. Through Bayes’ theorem, BNNs infer a posterior distribution over weights, allowing predictions to reflect epistemic uncertainty - the uncertainty arising from limited data and model knowledge. However, the exact posterior is analytically intractable for deep models, motivating approximate inference methods such as variational inference \citep{graves2011practical, blundell2015weight} and Monte Carlo dropout \citep{gal2016dropout}. Despite their appeal, these approaches may yield biased or overconfident posteriors due to restrictive variational families \citep{hernandez2015probabilistic, osband2023epistemic}, often resulting in over-smoothed predictive distributions.

An alternative paradigm for probabilistic modeling is the Mixture Density Network (MDN), introduced by \cite{bridle1990probabilistic} and developed further by \cite{jacobs1991adaptive}. Unlike BNNs, which encode uncertainty through distributions over parameters, MDNs model the conditional output density $p(y\vert x)$ directly as a weighted mixture of distributions (often Gaussians). The network outputs the parameters of the mixture - weights, means, and variances—thus allowing it to represent multimodal and heteroscedastic conditional structures \citep{bishop1994mixture, bishop1995neural}. This makes MDNs particularly effective in problems where multiple plausible outcomes exist for a single input, such as inverse problems, motion prediction, or medical diagnosis \citep{graves2013generating, laina2016deeper, tagasovska2019single}.

While both BNNs and MDNs provide probabilistic predictions, they emphasize different uncertainty sources. BNNs model epistemic uncertainty due to limited knowledge about the parameters. MDNs model aleatoric uncertainty and multimodality inherent in the data-generating process. Consequently, a systematic comparison of these approaches yields valuable insight into their complementary roles in uncertainty quantification.

Several recent directions have deepened the study of probabilistic neural modeling within both variational and likelihood-based frameworks. Variational methods have evolved through richer posterior families such as normalizing flows \citep{rezende2015variational, Louizos2017} and scalable gradient estimators \citep{Kingma2015, Wen2018}, significantly improving approximation flexibility. Concurrently, likelihood-based approaches have advanced through mixture and density modeling \citep{bishop1994mixture, bishop1995neural, graves2013generating, tagasovska2019single}, emphasizing expressive output distributions and data-driven uncertainty. These developments reflect two complementary paradigms - posterior-based inference and likelihood-based density estimation - whose relative theoretical properties remain insufficiently understood. This paper addresses that gap through a unified theoretical and empirical comparison of Bayesian Neural Networks and Mixture Density Networks.

In this paper, we conduct a comprehensive empirical and theoretical comparison of Bayesian Neural Networks and Mixture Density Networks for modeling nonlinear, potentially multimodal data. Both models are implemented using the PyTorch framework and are evaluated on synthetic datasets, generated from nonlinear functions with additive Gaussian noise, designed to exhibit multimodality and heteroscedasticity; and real-world data, namely the RSNA Pediatric Bone Age Challenge (2017) radiographic dataset, where uncertainty-aware prediction is crucial for clinical decision support. Our contributions are threefold: (i) Empirical comparison – We assess BNNs and MDNs on predictive calibration, uncertainty quantification, and multimodal representation using visualization and Kullback–Leibler (KL) divergence metrics, (ii) Theoretical analysis – We derive explicit approximation and estimation error bounds for both architectures, demonstrating that MDNs achieve faster KL-convergence rates under standard Hölder smoothness conditions, while BNNs incur additional terms due to prior and variational approximation, and (iii) Practical insights – We show that MDNs outperform BNNs in capturing multimodal outputs and adaptive uncertainty, whereas BNNs provide more interpretable epistemic uncertainty when data are scarce.

These findings integrate and extend insights from previous studies on Bayesian deep learning \citep{neal1996bayesian, blundell2015weight}, mixture modeling \citep{mclachlan2000finite}, and PAC-Bayesian generalization bounds \citep{mcallester1998some, catoni2012challenging}. The resulting framework clarifies theoretical trade-offs between posterior-based and likelihood-based uncertainty modeling, providing guidance for uncertainty-aware neural modeling in nonlinear systems.

The rest of this article is organized as follows. Section~\ref{prelim} introduces preliminary foundations of probabilistic neural modeling, including key divergence measures and sources of uncertainty. Section~\ref{bnn_limit} examines the Bayesian Neural Network framework, outlining its formulation, inference process, and main limitations. Section~\ref{mdn_section} presents the Mixture Density Network approach, highlighting its likelihood-based structure and theoretical advantages in modeling multimodal and heteroscedastic data. The prediction error bounds are compared in Section~\ref{theory}. Section~\ref{sim} reports empirical evaluations on both synthetic, and a real-world analysis is presented in Section~\ref{real_dat_analysis} . Finally, Section~\ref{discuss} concludes this article with a discussion.

\section{Preliminaries}\label{prelim}
Let $P$ and $Q$ denote two continuous probability distributions over $\mathcal{Y}$. The Kullback–Leibler (KL) divergence is defined as
\[
D_{\mathrm{KL}}(P\|Q) = \int \log\!\left(\frac{p(y)}{q(y)}\right) p(y)\,dy,
\]
and measures the discrepancy between two probability laws. Relatedly, the Rényi divergence of order $\alpha>0$, $\alpha\neq1$, is defined by
\[
D_\alpha(P\|Q)=\frac{1}{\alpha-1}\log\!\left(\int p(y)^{\alpha} q(y)^{1-\alpha}\,dy\right),
\]
which generalizes KL divergence as $\alpha\!\to\!1$. These divergences underpin much of the theoretical analysis of probabilistic neural models, especially in bounding approximation errors between the true data-generating distribution and model-implied predictive densities.

Uncertainty in neural models can broadly be classified into two categories:
\begin{itemize}
    \item\textbf{Epistemic uncertainty}, arising from limited knowledge of model parameters or data scarcity. This is typically captured via Bayesian inference over weights, as in BNNs.
    \item \textbf{Aleatoric uncertainty}, stemming from inherent stochasticity or multimodality in the data-generating process. MDNs are explicitly designed to capture this type through mixture-based likelihoods.
\end{itemize}

Throughout this paper, we consider the nonlinear regression setting where data $\mathcal{D}=\{(x_i,y_i)\}_{i=1}^n$ are generated according to an unknown stochastic process $Y=f^\ast(X)+\varepsilon$, with $f^\ast$ smooth and $\varepsilon$ heteroscedastic. The objective is to estimate the conditional density $p(y|x)$ and quantify uncertainty in predictions.

\section{Limitations of Bayesian Neural Networks}\label{bnn_limit}
BNNs replace deterministic weights with random variables, inducing a posterior distribution $p(\mathbf{w}|\mathcal{D})$ over network parameters. Predictive inference marginalizes over this posterior:
\[
p(y|x,\mathcal{D})=\int p(y|x,\mathbf{w})\,p(\mathbf{w}|\mathcal{D})\,d\mathbf{w}.
\]
Since exact inference is intractable, approximate methods - most notably variational inference - seek a tractable surrogate distribution $q_\phi(\mathbf{w})$ minimizing $D_{\mathrm{KL}}(q_\phi\|\!p)$. This yields the \emph{Evidence Lower Bound} (ELBO):
\[
\mathcal{L}_{\mathrm{ELBO}} = \mathbb{E}_{q_\phi(\mathbf{w})}[\log p(\mathcal{D}|\mathbf{w})]
- D_{\mathrm{KL}}(q_\phi(\mathbf{w})\,\|\,p(\mathbf{w})).
\]

A large and active literature has developed around variational approaches for scalable Bayesian neural networks, aiming to reduce approximation bias and improve posterior expressivity. Early work formulated variational optimization for neural network weights using stochastic gradient methods and the reparameterization trick \citep{graves2011practical,Kingma2013}, while the \emph{Bayes by Backprop} framework extended these ideas specifically to deep networks \citep{blundell2015weight}. Subsequent research proposed richer variational families and more expressive posteriors, such as normalizing-flow and multiplicative-flow distributions \citep{Rezende2015,Louizos2017}, and introduced local reparameterization tricks and variance-reduction schemes for efficient mini-batch training \citep{Kingma2015}. Complementary deterministic approximations, including probabilistic backpropagation and the Flipout estimator, trade off scalability and variance in different ways \citep{Hernandez2015,Wen2018}. Non-parametric and particle-based approaches, notably Stein variational gradient descent, approximate the posterior without specifying a variational density \citep{Liu2016}. Collectively, these advances mitigate several limitations of classical mean-field variational BNNs and motivate hybrid architectures that combine expressive variational posteriors with likelihood-flexible output layers to capture both epistemic and aleatoric uncertainty.

Although conceptually elegant, ELBO optimization introduces several limitations:

\begin{enumerate}
    \item \textbf{Variational bias.} Restrictive variational families (e.g., Gaussian mean-field) often underestimate posterior uncertainty, producing overconfident predictions.
    \item \textbf{Training instability.} ELBO optimization requires careful balancing between likelihood and KL terms; poor scaling can cause mode collapse or posterior drift.
    \item \textbf{Computational burden.} Sampling-based gradient estimation and large parameter spaces render BNN training computationally expensive compared to deterministic networks.
    \item \textbf{Limited expressiveness.} BNNs primarily capture epistemic uncertainty but struggle to represent multimodal conditional distributions inherent in stochastic processes.
\end{enumerate}

To illustrate, we consider a synthetic dataset generated by
\[
Y = \sin(2\pi X) + 0.5\cos(6\pi X) + \varepsilon, \quad \varepsilon\sim\mathcal{N}(0,\sigma^2),
\]
where $X\sim\mathrm{Unif}(0,1)$. This data exhibits multimodal and nonlinear patterns. A two-layer BNN trained by ELBO minimization captures average trends but fails to resolve multimodal regions - reflecting its inability to express multiple plausible outputs for the same input.

These limitations have practical consequences in real-world tasks such as medical imaging or autonomous perception, where uncertainty estimation must encompass both epistemic and aleatoric components. The next section introduces the Mixture Density Network as a flexible alternative.

\section{Mixture Density Networks}\label{mdn_section}
A Mixture Density Network (MDN) \citep{bishop1994mixture,bishop1995neural} combines a standard feed-forward neural architecture with a parametric mixture model, typically Gaussian. Instead of outputting a single deterministic value, the MDN outputs the parameters of a conditional mixture distribution:
\[
p(y|x;\Theta) = \sum_{k=1}^K \pi_k(x)\,\mathcal{N}\!\left(y\,\middle|\,\mu_k(x),\,\sigma_k^2(x)\right),
\]
where $\pi_k(x)$, $\mu_k(x)$, and $\sigma_k^2(x)$ are network-generated mixture weights, means, and variances, respectively, and $\sum_k \pi_k(x)=1$. This formulation enables the MDN to represent multimodal, heteroscedastic, and asymmetric conditional relationships directly.

Training proceeds by maximizing the log-likelihood of the observed data:
\[
\mathcal{L}_{\mathrm{MDN}} = \sum_{i=1}^n \log p(y_i|x_i;\Theta),
\]
using gradient-based optimization. Because MDNs model the conditional density explicitly, they avoid the variational approximations inherent in BNNs and can capture both aleatoric and structural uncertainty.

In summary, BNNs encode uncertainty through parameter distributions, while MDNs directly model the conditional output density. The former is posterior-based (epistemic), the latter likelihood-based (aleatoric). These complementary formulations motivate the theoretical comparison that follows.
\section{Theoretical results}\label{theory}

In this section, we compare the performance of MDN and BNN in terms of their prediction accuracy. We begin by stating a few assumptions on the model class, choice of prior, etc. 

\paragraph{Assumptions.}
\begin{itemize}
  \item[(A1)] (\textbf{True model}). The true conditional density $f^\ast(y\mid x)$ on $\mathbb{R}$ given $x\in\mathcal{X}\subset\mathbb{R}^d$ admits an $M$-component Gaussian mixture representation
  \[
    f^\ast(y\mid x)=\sum_{m=1}^M \pi_m(x)\,\phi(y;\mu_m(x),\sigma_m^2(x)),
  \]
  and the parameter functions $\pi_m,\mu_m,\sigma_m$ are $s$-Hölder continuous on $\mathcal{X}$.
  \item[(A2)] (\textbf{Boundedness and non-degeneracy}). There exist constants $\varepsilon\in(0,1/2)$ and $0<\sigma_{\min}<\sigma_{\max}<\infty$ such that, for all $x\in\mathcal{X}$ and $m$,
  \[
    \pi_m(x)\in[\varepsilon,1-\varepsilon],\qquad \sigma_m(x)\in[\sigma_{\min},\sigma_{\max}].
  \]
  \item[(A3)] (\textbf{Model class / network parametrization}). For integers $K\ge M$ and width parameter $n$, let $\mathcal{F}_{n,K}$ denote the class of $K$-component Gaussian mixtures whose parameter functions $(\pi_k,\mu_k,\sigma_k)$ are implemented by ReLU neural networks of width at most $n$. Let $C(n,K,d)$ denote a suitable complexity measure of $\mathcal{F}_{n,K}$ (for example the pseudo-dimension or a log-covering-number proxy).
  \item[(A4)] (\textbf{Estimator}). The estimator $\widehat f_{n,K}$ is an empirical-risk minimizer (ERM) over $\mathcal{F}_{n,K}$ using the negative log-likelihood on $N$ i.i.d. samples $(X_i,Y_i)_{i=1}^N$.
  \item[(A5)] (\textbf{Variational family / prior }). For the Bayesian analysis we assume a prior $\pi(w)$ on network weights and a variational family $\mathcal{Q}$ that is rich enough to place mass concentrated near network weights realizing good approximating networks (this is standard; see discussion in the main text).
\end{itemize}

We first present the auxiliary lemmas required by the proofs and then give the proofs of the two main theorems.

\begin{lemma}[Finite-mixture identity]
\label{lem:mixture-approx}
Under assumptions (A1)–(A2), if the number of mixture components $K \ge M$, then the true conditional density $f^*(y \mid x)$ can be represented exactly by a $K$-component Gaussian mixture. Consequently,
\[
\sup_{x \in \mathcal{X}} 
  \mathrm{KL}\big(f^*(\cdot \mid x) \,\|\, f_K(\cdot \mid x)\big) = 0.
\]
\end{lemma}

\begin{proof}
The proof is deferred to \Cref{proof:lem:mixture-approx}
\end{proof}

\begin{lemma}[ReLU approximation of H\"older functions]
\label{lem:relu-approx}
Let $g:\mathcal{X}\to\mathbb{R}$ be $s$–Hölder continuous on a compact domain $\mathcal{X}\subset\mathbb{R}^d$, i.e.
$|g(x)-g(x')| \le L\|x-x'\|_\infty^s$ for all $x,x'$. Then, for every integer $n\ge1$, there exists a ReLU network $\tilde g_n$
of width at most $n$ (and depth depending on $s,d$) such that
\[
\|g-\tilde g_n\|_{\infty,\mathcal{X}} \le C_{\mathrm{app}}\, n^{-s/d},
\]
where $C_{\mathrm{app}}>0$ depends only on $s,d,L$, and $\mathrm{diam}(\mathcal{X})$.
\end{lemma}

\begin{proof}
See \Cref{proof:relu-approx}.
\end{proof}

\begin{lemma}[Sup-norm parameter perturbation $\Rightarrow$ KL control]
\label{lem:sup-to-kl}
Let $f$ and $\tilde f$ be $K$–component Gaussian mixtures satisfying (A2) with parameter functions
$(\pi_k,\mu_k,\sigma_k)$ and $(\tilde\pi_k,\tilde\mu_k,\tilde\sigma_k)$.
Define
\[
\varepsilon_{\mathrm{par}}
  := \max_k\{\|\pi_k-\tilde\pi_k\|_\infty,\|\mu_k-\tilde\mu_k\|_\infty,\|\sigma_k-\tilde\sigma_k\|_\infty\}.
\]
Then, for sufficiently small $\varepsilon_{\mathrm{par}}$, there exists $C_{\mathrm{KL}}>0$ depending only on
$\sigma_{\min},\sigma_{\max},\varepsilon$ such that
\[
\sup_{x\in\mathcal{X}} \mathrm{KL}\!\big(f(\cdot\mid x)\,\|\,\tilde f(\cdot\mid x)\big)
 \le C_{\mathrm{KL}}\,K\,\varepsilon_{\mathrm{par}}^2.
\]
\end{lemma}

\begin{proof}
The proof is given in \Cref{proof:sup-to-kl}.
\end{proof}



\begin{lemma}[ERM concentration / estimation error]
\label{lem:erm}
Let $\mathcal{F}_{n,K}$ be the class of $K$–component Gaussian mixtures with ReLU parameter functions of width $\le n$.
Assume the covering number satisfies
$\log N(\epsilon,\mathcal{F}_{n,K},\|\cdot\|_\infty)\le C(n,K,d)\log(1/\epsilon)$.
Let $\hat f_{n,K}$ denote the empirical risk minimizer under the negative log-likelihood.
Then, for any $\delta\in(0,1)$, with probability at least $1-\delta$,
\[
\sup_{f\in\mathcal{F}_{n,K}}
  |P_N\log f - P\log f|
  \le C_3 \sqrt{\frac{C(n,K,d)+\log(1/\delta)}{N}},
\]
and consequently,
\[
\mathrm{KL}(f_{n,K}\|\hat f_{n,K})
  \le C_3 \sqrt{\frac{C(n,K,d)+\log(1/\delta)}{N}}.
\]
\end{lemma}

\begin{proof}
See \Cref{proof:erm}

\end{proof}

\begin{lemma}[PAC-Bayes inequality]
\label{lem:pac-bayes}
Let $\pi$ be a prior on network weights and $q$ any posterior distribution.
Define $p_q(y\mid x)=\mathbb{E}_{w\sim q}[p(y\mid x,w)]$.
Then for any $\delta\in(0,1)$, with probability at least $1-\delta$ over $D=\{(X_i,Y_i)\}_{i=1}^N$,
\[
\mathbb{E}_X \mathrm{KL}\!\big(f^*(\cdot\mid X)\,\|\,p_q(\cdot\mid X)\big)
 \le \frac{1}{N}\sum_{i=1}^N \mathbb{E}_{w\sim q}[-\log p(Y_i\mid X_i,w)]
   + \frac{\mathrm{KL}(q\|\pi)+\log(1/\delta)}{N}.
\]
\end{lemma}

\begin{proof}

The proof is given in \Cref{proof:pac-bayes}.

\end{proof}

\begin{theorem}\label{thm:mdn}
(MDN KL convergence for exact mixture) Under (A1)–(A4), let $f^*$ be the true conditional density and $\hat f_{n,K}$ the ERM over $\mathcal{F}_{n,K}$ with $K\ge M$.
Then there exist constants $C_2,C_3>0$ such that, with probability at least $1-\delta$,
\[
\mathrm{KL}\!\big(f^*\|\hat f_{n,K}\big)
 \le C_2 K n^{-2s/d}
   + C_3 \sqrt{\frac{C(n,K,d)+\log(1/\delta)}{N}}.
\]
\end{theorem}

\begin{proof}
Decompose
\[
\mathrm{KL}(f^*\|\hat f_{n,K})
  = \mathrm{KL}(f^*\|f_K)
   +\mathrm{KL}(f_K\|f_{n,K})
   +\mathrm{KL}(f_{n,K}\|\hat f_{n,K}).
\]
The first term vanishes by \Cref{lem:mixture-approx}.
\Cref{lem:relu-approx} combined with \Cref{lem:sup-to-kl} gives $\mathrm{KL}(f_K\|f_{n,K})\le C_2 K n^{-2s/d}$.
\Cref{lem:erm} yields the empirical estimation bound\\
$\mathrm{KL}(f_{n,K}\|\hat f_{n,K})\le C_3\sqrt{[C(n,K,d)+\log(1/\delta)]/N}$.
Summing completes the proof.
\end{proof}

\begin{theorem}\label{thm:bnn}
(BNN / PAC-Bayes posterior predictive bound) Under (A1)–(A5), let $\pi$ be a prior on network weights and $\mathcal{Q}$ a variational family satisfying (A5).
Then there exists $q^*\in\mathcal{Q}$ such that, for any $\delta\in(0,1)$,
\[
\mathbb{E}_X\mathrm{KL}\!\big(f^*(\cdot\mid X)\,\|\,p_\pi(\cdot\mid X,D)\big)
 \le C_2 K n^{-2s/d}
   + \frac{\mathrm{KL}(q^*\|\pi)+\log(1/\delta)}{N}.
\]
\end{theorem}

\begin{proof}
By (A5), choose $q^*\in\mathcal{Q}$ concentrated around weights realizing the approximating network $f_{n,K}$.
Applying \Cref{lem:pac-bayes} with $q=q^*$ gives
\[
\mathbb{E}_X\mathrm{KL}\!\big(f^*(\cdot\mid X)\,\|\,p_{q^*}(\cdot\mid X)\big)
 \le \frac{1}{N}\!\sum_{i=1}^N \mathbb{E}_{w\sim q^*}[-\log p(Y_i\mid X_i,w)]
   +\frac{\mathrm{KL}(q^*\|\pi)+\log(1/\delta)}{N}.
\]
Because $q^*$ is concentrated in a small neighborhood of the weights generating $f_{n,K}$,
$\mathbb{E}_{w\sim q^*}[-\log p(Y_i\mid X_i,w)] = -\log f_{n,K}(Y_i\mid X_i) + o(1)$.
Taking expectations under $f^*$ and invoking Lemmas 2–3 yields the approximation bound
$\mathbb{E}_X\mathrm{KL}(f^*(\cdot\mid X)\|f_{n,K}(\cdot\mid X)) \le C_2 K n^{-2s/d}$.
Combining terms gives the desired result.
\end{proof}

Finally, we conclude this section with the following remarks.

(i) \Cref{thm:mdn} and \Cref{thm:bnn} together clarify the distinct statistical behaviors of MDNs and BNNs. The MDN bound combines an approximation term of order $K n^{-2s/d}$ with an estimation term of order $\sqrt{C(n,K,d)/N}$, leading to consistency and fast convergence under standard smoothness and boundedness assumptions. In contrast, the BNN bound inherits an additional $\mathrm{KL}(q^\ast\|\pi)/N$ term, reflecting the impact of prior mismatch and variational approximation on generalization.
  
 (ii) These results formally explain the empirical findings: MDNs tend to recover multimodal or heteroscedastic conditional densities more accurately, while BNNs trained via variational inference may exhibit over-smoothed or inflated uncertainty estimates when the variational family is restrictive. The theoretical gap between the two methods thus directly corresponds to the practical performance gap observed in simulation.
  
(iii) The dependence on $(n,K,d)$ shows that the expressive capacity of the MDN architecture controls approximation bias, whereas sample size $N$ governs estimation error. For BNNs, even large models cannot eliminate the bias introduced by limited variational flexibility or inappropriate priors, emphasizing the importance of posterior expressivity in Bayesian deep learning.
  
(iv) The $K n^{-2s/d}$ term highlights that approximation errors decay quadratically with respect to the network’s functional approximation accuracy, confirming that smoother target conditionals (larger $s$) or wider networks (larger $n$) yield faster convergence.
  
(v) Overall, the theoretical analysis establishes that MDNs offer a direct and statistically efficient route to conditional density estimation, while BNNs trade statistical efficiency for a probabilistic interpretability that depends critically on the quality of the variational posterior. This delineation provides a principled explanation for the simulation outcomes and guides model choice in practice.


\section{Simulations}\label{sim}

This section reports a controlled simulation comparing a Bayesian Neural Network (BNN) trained via variational inference (VI) and a Mixture Density Network (MDN). The implementation follows the accompanying PyTorch/NumPy script, using 3000 training epochs for both models, a learning rate of $10^{-3}$, and the Adam optimizer. All experiments were conducted with fixed random seeds for reproducibility.

\subsection*{Data-generating processes}

For each case, $n=800$ data points are generated with additive Gaussian noise $\varepsilon_i \sim \mathcal{N}(0,0.1^2)$, $i=1,\ldots,n$. The four cases correspond to distinct nonlinear relationships between $X$ and $Y$:

\begin{itemize}
\item \textbf{Case A (Cubic):} $f(x)=x^3$, representing a smooth monotone nonlinear trend.

\item \textbf{Case B (Piecewise):}
\[
f(x)=
\begin{cases}
x^2, & x<0,\\
-1.5x+0.3, & x\ge 0,
\end{cases}
\]
exhibiting a sharp change in slope at $x=0$.

\item \textbf{Case C (Bimodal):}
\[
Y\mid X=x \sim 0.5\,\mathcal{N}(x+1,0.1^2) + 0.5\,\mathcal{N}(-x-1,0.1^2),
\]
generating two overlapping Gaussian modes and a multi-modal conditional distribution.

\item \textbf{Case D (Sinusoidal):} $f(x)=\sin(3x)+0.3\sin(9x)$, representing a highly oscillatory function with multiple local extrema.
\end{itemize}

Inputs $X_i$ are sampled uniformly from $\mathrm{Unif}[-3,3]$, and data are split into 80\% training and 20\% testing. A dense grid of 500 equally spaced points in $[-3,3]$ is used for evaluation and visualization.

\subsection*{Models and training}

\paragraph{Bayesian Neural Network (VI).}
The BNN consists of two stochastic Bayesian linear layers with $\tanh$ activation:
\[
1 \rightarrow 50 \rightarrow 1,
\]
where all weights and biases are modeled with Gaussian variational posteriors $q(\theta)=\mathcal{N}(\mu,\sigma^2)$. The network is trained by minimizing the negative evidence lower bound (ELBO):
\[
\mathcal{L}_{\text{VI}}(\theta)
= -\mathbb{E}_{q(\theta)}[\log p(Y\mid X,\theta)] + \frac{1}{n}\mathrm{KL}\big(q(\theta)\,\Vert\,p(\theta)\big),
\]
with a standard normal prior $p(\theta)=\mathcal{N}(0,1)$. During inference, $T=200$ Monte Carlo forward passes are performed to approximate the predictive mean and uncertainty:
\[
\hat{\mu}_{\text{BNN}}(x)=\frac{1}{T}\sum_{t=1}^T f_t(x), \qquad
\hat{\sigma}_{\text{BNN}}(x)=\sqrt{\frac{1}{T}\sum_{t=1}^T \big(f_t(x)-\hat{\mu}_{\text{BNN}}(x)\big)^2}.
\]

\paragraph{Mixture Density Network (MDN).}
The MDN parameterizes $p(y\mid x)$ as a Gaussian mixture with $K=5$ components. The architecture includes one hidden layer of width 50 with $\tanh$ activation, followed by output heads for mixture weights, component means, and standard deviations. The predictive mean and variance are given by:
\[
\mu_{\text{MDN}}(x)=\sum_{k=1}^K \pi_k(x)\mu_k(x),\quad
\operatorname{Var}_{\text{MDN}}(x)=\sum_{k=1}^K \pi_k(x)\big(\sigma_k^2(x)+\mu_k^2(x)\big)-\mu_{\text{MDN}}^2(x).
\]

\begin{figure}[!ht]
\centering
\includegraphics[width=1\textwidth]{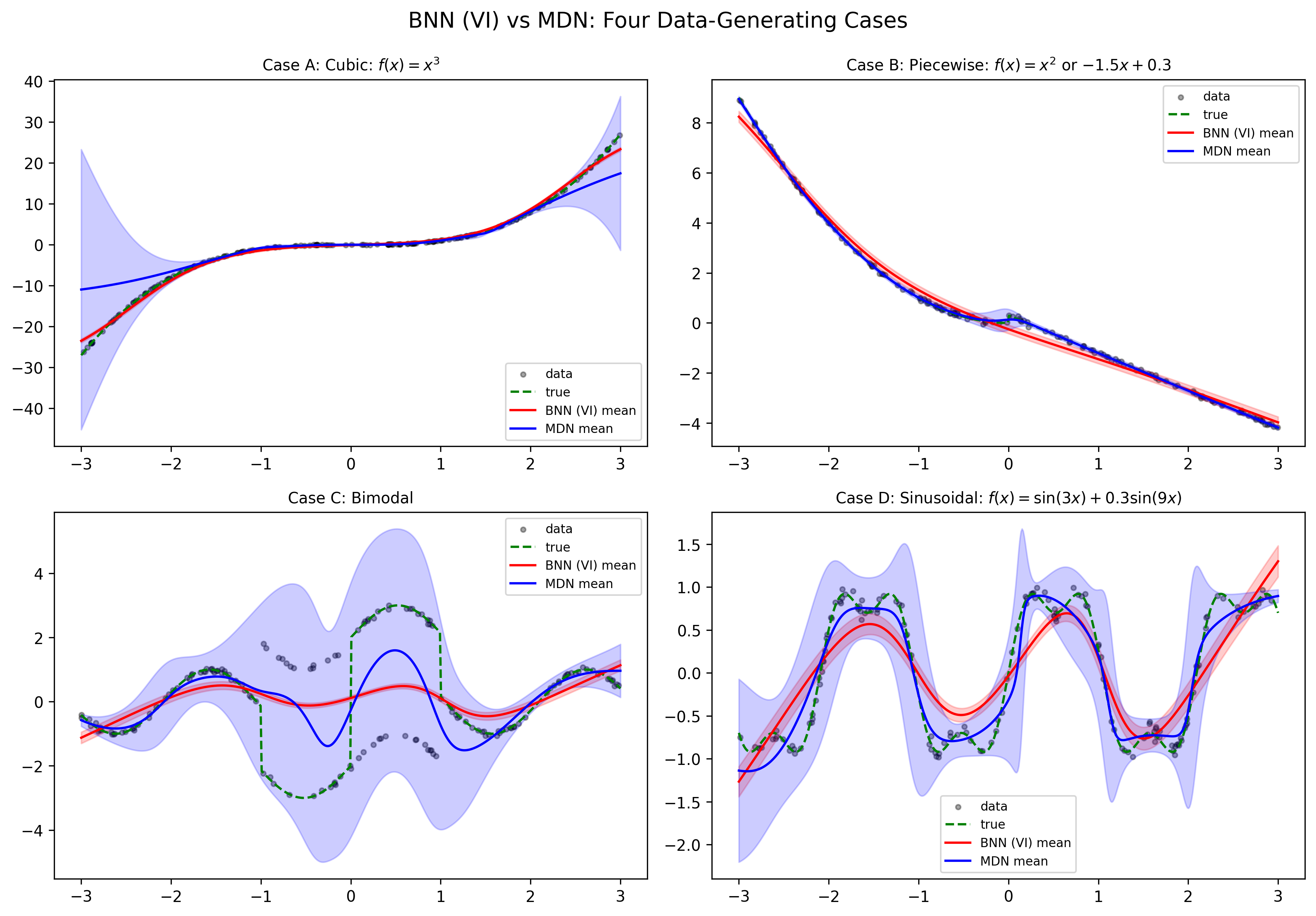}
\caption{\textbf{Predictive comparison between BNN (VI) and MDN.}
Each panel corresponds to one data-generating function. The dashed green curve shows the ground truth $f(x)$, while the shaded bands represent $\pm2$ standard deviation predictive intervals.
The BNN exhibits wide, homoscedastic uncertainty, whereas the MDN adapts its predictive variance to data complexity and multimodality.}
\label{fig:bnn-mdn}
\end{figure}

\subsection*{Predictive performance}

Table~\ref{tab:nll} reports the average test-set negative log-likelihoods (NLL) per sample for each simulation setting. 
For the BNN, $T=200$ Monte Carlo samples are used for estimating the predictive likelihood. 
For the MDN, the exact mixture likelihood is computed analytically from the learned component parameters.

\begin{table}[!ht]
\centering
\caption{Per-sample test-set negative log-likelihoods (NLL) for BNN (VI) and MDN across four simulation settings. Lower values correspond to better-calibrated predictive densities. Negative values can occur when the model assigns probability densities greater than one, which is valid in continuous settings.}
\label{tab:nll}
\begin{tabular}{@{}lcc@{}}
\toprule
\textbf{Case} & \textbf{BNN (VI)} & \textbf{MDN} \\ \midrule
A (Cubic) & 1.0817 & \textbf{-0.1946} \\
B (Piecewise) & 3.5625 & \textbf{1.2553} \\
C (Bimodal) & 37.3510 & \textbf{0.5883} \\
D (Sinusoidal) & 20.4724 & \textbf{-0.1514} \\ \bottomrule
\end{tabular}
\end{table}

The results highlight the expressive advantage of the MDN in predictive calibration and density estimation:

\begin{itemize}
    \item \textbf{Cubic and Piecewise:} Even in smooth regimes, the BNN (VI) tends to produce diffuse, underconfident posteriors, leading to higher NLL. The MDN, benefiting from its mixture representation, produces sharper and better-calibrated likelihoods.
    \item \textbf{Bimodal:} The BNN’s unimodal Gaussian output cannot represent multimodal targets, resulting in a dramatic deterioration in likelihood. The MDN accurately models both modes, achieving a much lower NLL.
    \item \textbf{Sinusoidal:} The BNN smooths over high-frequency oscillations and inflates predictive variance, whereas the MDN flexibly adapts mixture components to capture the nonlinear periodicity.
\end{itemize}

Overall, the MDN provides more expressive and better-calibrated predictive distributions across all four scenarios.
While the variational BNN captures epistemic uncertainty through posterior sampling, its Gaussian observation model limits its ability to represent heteroscedasticity and multimodality.
In contrast, the MDN directly parameterizes complex conditional densities, achieving both flexibility and superior likelihood-based performance.


\section{Real Data Analysis}\label{real_dat_analysis}

We conducted a real-world evaluation using the \textbf{RSNA Pediatric Bone Age Challenge 2017} dataset, which consists of pediatric hand radiographs accompanied by expert-provided bone age annotations. The predictive task is to estimate bone age (in months) from radiographic images, a clinically important diagnostic measure in pediatric endocrinology.  

This dataset embodies challenges typical of medical imaging: annotation subjectivity, inherent image noise, and complex non-linear relationships between image features and bone age. These characteristics make it a suitable benchmark for assessing uncertainty-aware models such as Bayesian Neural Networks (BNNs) and Mixture Density Networks (MDNs). In this analysis, we utilized the \textbf{full dataset}, partitioned into training and validation subsets for model development and evaluation.  

\subsection{Preprocessing and Normalization}
To standardize inputs and improve model stability, the following preprocessing steps were applied:
\begin{itemize}
    \item \textbf{Image resizing:} All radiographs were resized to $128 \times 128$ pixels.  
    \item \textbf{Pixel normalization:} Images were normalized using ImageNet statistics (mean = [0.485, 0.456, 0.406], std = [0.229, 0.224, 0.225]).  
    \item \textbf{Label normalization:} Bone age labels (months) were standardized to zero mean and unit variance, using statistics computed only from the training set. This ensured stable optimization in the regression setting.  
\end{itemize}

\subsection{Model Architectures and Training}
Both models employed a \textbf{CNN backbone} with four convolutional layers (each followed by ReLU and max pooling) to extract image features. The resulting feature vector was then passed to model-specific prediction heads.  

\textbf{BNN:} The features were fed into two Bayesian linear layers, where weights and biases followed Gaussian distributions. Predictions were distributions rather than point estimates. Training minimized the \textbf{Evidence Lower Bound (ELBO)} with a KL-weight of 0.001 to balance likelihood and prior regularization.  

\textbf{MDN:} Features were passed into a fully connected layer branching into three heads, predicting mixture weights ($\pi$), means ($\mu$), and standard deviations ($\sigma$) of \textbf{3 Gaussian components} ($N=3$). Training maximized the log-likelihood of the data under the predicted mixture distribution.  

Both models were trained for \textbf{60 epochs} using Adam (learning rate = $1 \times 10^{-3}$).  

\subsection{Results}
Performance was evaluated on the held-out validation set. For each model, we report the mean prediction and standard deviation of predictive uncertainty.  
\begin{itemize}
    \item \textbf{BNN:} Predictions were obtained by repeatedly sampling network weights from their learned posterior.  
    \item \textbf{MDN:} Predictions were sampled from the learned Gaussian mixture distribution.  
\end{itemize}

Figure~\ref{fig:predictive_distributions_real_data} presents predicted vs. true bone ages, with error bars denoting $\pm 1$ standard deviation. The dashed diagonal represents perfect prediction.  

\begin{figure}[h!]
\centering
\includegraphics[width=1.0\textwidth]{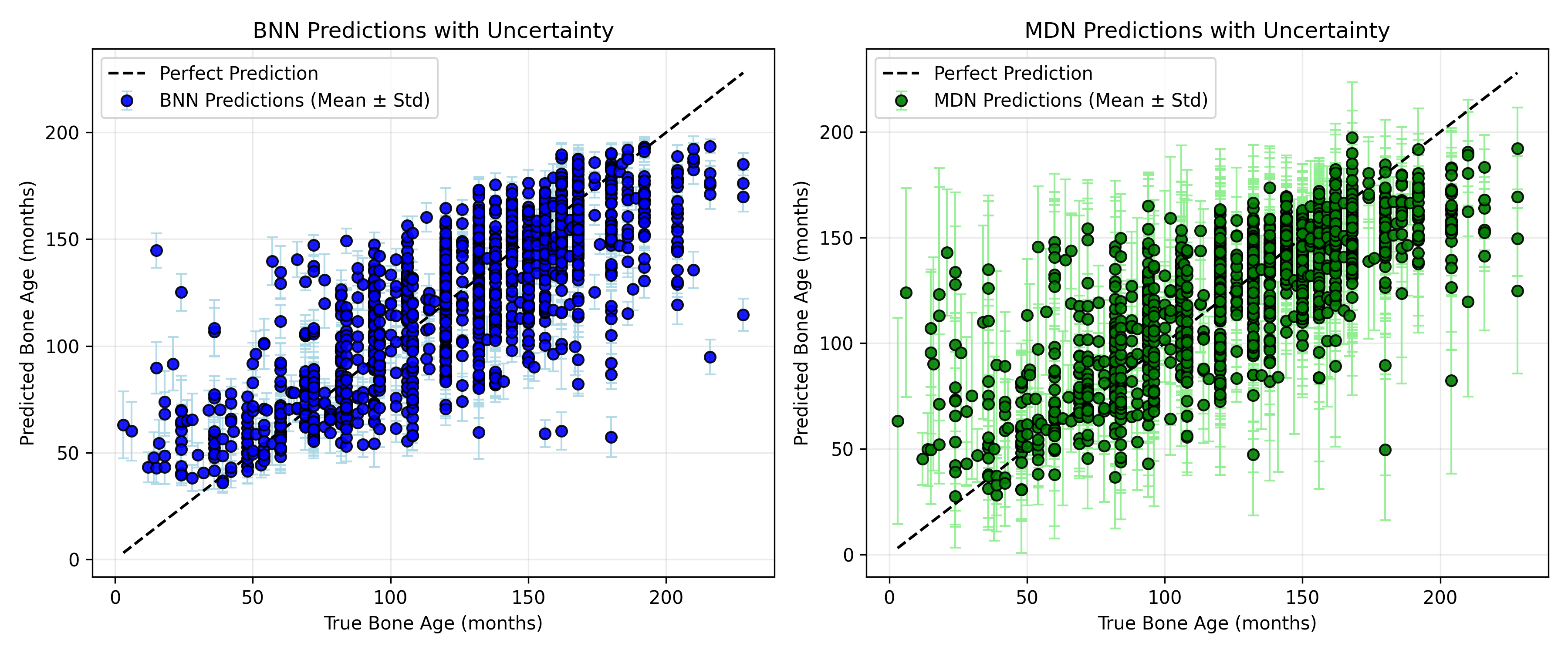}
\caption{\textbf{Predicted vs. True Bone Age on Validation Data.} Left: BNN predictions. Right: MDN predictions. Points represent mean predicted bone age; error bars indicate $\pm 1$ standard deviation. The dashed diagonal corresponds to perfect prediction.}
\label{fig:predictive_distributions_real_data}
\end{figure}

\subsubsection*{BNN Analysis (Left)}
The BNN captures the overall trend in bone age prediction, with predictions clustering around the perfect-prediction line. Uncertainty is heteroscedastic: smaller in regions well-represented by the training data and larger where data are sparse or ambiguous. This reflects the BNN’s ability to capture \textbf{epistemic uncertainty}. However, some outliers remain, where predictions deviate significantly with wide uncertainty bounds, likely due to posterior approximation challenges or highly complex image features.  

\subsubsection*{MDN Analysis (Right)}
The MDN also follows the diagonal trend but exhibits broader and more uniformly distributed uncertainty intervals across the age spectrum. Unlike the BNN, the MDN explicitly models \textbf{aleatoric uncertainty} and multimodality. While this plot summarizes predictions using mean $\pm$ std, the underlying Gaussian mixtures can represent multiple plausible bone ages for a single image. The broader but steadier uncertainty estimates suggest the MDN effectively captures inherent variability in bone development.  

\subsubsection*{Comparative Insights}
\begin{itemize}
    \item \textbf{Uncertainty quantification:}  
    BNNs emphasize epistemic uncertainty from data scarcity and model limitations, while MDNs emphasize aleatoric uncertainty and can capture multimodal predictions.  
    \item \textbf{Predictive behavior:}  
    Both models align reasonably well with the true bone age trend. BNN uncertainties expand in less confident regions, whereas MDN uncertainties remain broader but steadier across all ages.  
    \item \textbf{Clinical relevance:}  
    Uncertainty-aware predictions are essential in medical settings, where prediction confidence can guide expert review. MDNs may be particularly useful when genuine multimodality exists (e.g., ambiguous developmental patterns), whereas BNNs are valuable for quantifying confidence in data-limited scenarios.  
\end{itemize}

\noindent In summary, both BNNs and MDNs provide significant advantages over deterministic neural networks for medical imaging tasks such as bone age prediction. The choice between them depends on whether \textbf{epistemic uncertainty (BNNs)} or \textbf{aleatoric/multimodal uncertainty (MDNs)} is more critical for the intended application.


\section{Discussion}\label{discuss}

The comparative analysis reveals that Bayesian Neural Networks (BNNs) and Mixture Density Networks (MDNs) embody distinct yet complementary approaches to probabilistic modeling. BNNs provide a principled Bayesian treatment of parameter uncertainty, yielding interpretable measures of epistemic uncertainty that are especially valuable in low-data regimes. In contrast, MDNs, through their likelihood-based formulation, offer a flexible and efficient mechanism for modeling multimodal and heteroscedastic conditional relationships. Across both synthetic and empirical datasets, MDNs demonstrate stronger alignment with the true conditional density structure, achieving sharper and more adaptive uncertainty quantification. These findings highlight that the choice between BNNs and MDNs should depend on the dominant source of uncertainty in the application domain - epistemic versus aleatoric - and on the interpretability requirements of the task.

Despite their strengths, both models face limitations that point to several directions for future work. BNNs, while theoretically grounded, remain computationally demanding and sensitive to variational approximations, which can distort posterior uncertainty. MDNs, though capable of modeling rich distributions, may suffer from mode collapse or numerical instability when the number of mixture components is large. Furthermore, neither framework fully resolves the joint representation of epistemic and aleatoric uncertainty within a unified model. Future research should explore hybrid architectures that integrate Bayesian parameter inference with mixture-based output layers, as well as scalable inference schemes such as amortized or hierarchical variational methods to enhance both tractability and expressiveness in uncertainty-aware neural models.


\section{Appendix: Proofs of Theoretical Results}

This appendix provides detailed proofs for Lemmas~1--5 and Theorems~1--2 presented in the main text. Throughout, we denote by $f^*(y \mid x)$ the true conditional density of $Y$ given $X=x$, and by $\|\cdot\|_{\infty,\mathcal{X}}$ the supremum norm over $\mathcal{X}$. Constants denoted by $C, C_1, C_2, \ldots$ may vary from line to line but depend only on fixed problem parameters such as $s, d, \sigma_{\min}, \sigma_{\max},$ and $\varepsilon$.

\subsection{Proof of \Cref{lem:mixture-approx} (Finite-mixture identity)}
\label{proof:lem:mixture-approx}

\begin{proof}
By assumption (A1), $f^*(y \mid x) = \sum_{m=1}^M \pi_m(x)\phi(y; \mu_m(x),\sigma_m^2(x))$
for some $M$-component Gaussian mixture. Taking any $K \ge M$ and defining
$f_K(y \mid x) = \sum_{k=1}^K \pi_k(x)\phi(y; \mu_k(x),\sigma_k^2(x))$
with the first $M$ components identical to those of $f^*$ and setting the remaining
weights arbitrarily small so that $\sum_k \pi_k(x)=1$, we obtain $f_K=f^*$. Hence the KL divergence is zero. See also \cite{nguyen2013convergence} and Chapter 2 of \cite{mclachlan2000finite} for exact representation of mixtures.
\end{proof}


\subsection{Proof of \Cref{lem:relu-approx} (ReLU approximation of H\"older functions)}
\label{proof:relu-approx}

\begin{proof}
This follows from constructive approximation results for ReLU networks
\citep{yarotsky2017error, lu2017expressive}.
For every $s$–Hölder function $g$ there exists a ReLU network with width $\mathcal{O}(n)$ and depth $\mathcal{O}(\log n)$
such that $\|g-\tilde g_n\|_\infty = \mathcal{O}(n^{-s/d})$.
Applying this to each parameter function $(\pi_m,\mu_m,\sigma_m)$ in the mixture
yields the stated rate.
\end{proof}


\subsection{Proof of \Cref{lem:sup-to-kl} (Sup-norm parameter perturbation $\Rightarrow$ KL control)}
\label{proof:sup-to-kl}

\begin{proof}
(i) \textit{Component-level bound.}
For univariate Gaussians $p=N(\mu,\sigma^2)$ and $q=N(\tilde\mu,\tilde\sigma^2)$,
\[
\mathrm{KL}(p\|q)
  = \log\!\frac{\tilde\sigma}{\sigma}
    +\frac{\sigma^2+(\mu-\tilde\mu)^2}{2\tilde\sigma^2}-\tfrac12.
\]
A Taylor expansion around $(\tilde\mu,\tilde\sigma)=(\mu,\sigma)$ with $\sigma,\tilde\sigma\in[\sigma_{\min},\sigma_{\max}]$
gives $\mathrm{KL}(p\|q)\le C_0\big((\mu-\tilde\mu)^2+(\sigma-\tilde\sigma)^2\big)$.

(ii) \textit{Mixture-level bound.}
Let $f=\sum_k \pi_k p_k$ and $\tilde f=\sum_k \tilde\pi_k \tilde p_k$.
By convexity of KL and the decomposition inequality (Nguyen, 2013),
\[
\mathrm{KL}(f\|\tilde f)
 \le \mathrm{KL}(\pi\|\tilde\pi) + \sum_k \pi_k\,\mathrm{KL}(p_k\|\tilde p_k),
\]
where $\pi=(\pi_1,\ldots,\pi_K)$.
A Taylor expansion of $\mathrm{KL}(\pi\|\tilde\pi)$ around $\tilde\pi=\pi$ yields
$\mathrm{KL}(\pi\|\tilde\pi)\le (2\varepsilon)^{-1}\sum_k(\pi_k-\tilde\pi_k)^2$.
Combining both bounds yields
\[
\mathrm{KL}(f\|\tilde f)\le C_{\mathrm{KL}}\,K\,\varepsilon_{\mathrm{par}}^2.
\]
\end{proof}


\subsection{Proof of \Cref{lem:erm} (ERM concentration / estimation error)}
By empirical process theory \citep{van1996weak},
if $\log N(\epsilon,\mathcal{F},\|\cdot\|_\infty)\lesssim C\log(1/\epsilon)$, then
\[
\sup_{f\in\mathcal{F}} |(P_N-P)f|
   = \mathcal{O}_p\!\Big(\sqrt{\frac{C+\log(1/\delta)}{N}}\Big).
\]
Applying this to the uniformly bounded class $\{\log f : f\in\mathcal{F}_{n,K}\}$—boundedness ensured by (A2)—
yields the uniform deviation bound.  
Because $\hat f_{n,K}$ maximizes $P_N\log f$, the inequality
$P\log f_{n,K}-P\log\hat f_{n,K}\le\sup_f |(P_N-P)\log f|$
implies the stated KL deviation.
\label{proof:erm}

\subsection{Proof of \Cref{lem:pac-bayes} (PAC-Bayes inequality)}
\label{proof:pac-bayes}

\begin{proof}
The bound follows from the PAC–Bayesian theorem for bounded log-likelihood losses \citep{mcallester1998some,catoni2012challenging}:
for any prior $\pi$ and posterior $q$,
\[
\mathbb{E}_{X,Y}\mathbb{E}_{w\sim q}\ell(Y,X,w)
 \le \frac{1}{N}\sum_{i=1}^N\mathbb{E}_{w\sim q}\ell(Y_i,X_i,w)
   +\frac{\mathrm{KL}(q\|\pi)+\log(1/\delta)}{N},
\]
where $\ell(y,x,w)=-\log p(y\mid x,w)$.  
Writing $\mathbb{E}_{X,Y}\ell(Y,X,w)
 =\mathbb{E}_X\mathrm{KL}\big(f^*(\cdot\mid X)\|p(\cdot\mid X,w)\big)+H(f^*)$
and omitting the constant entropy term yields the claim.
\end{proof}




\medskip
\bibliographystyle{apacite}
	
\bibliography{reference}
\end{document}